\documentclass[12pt,final]{l4dc2020} 

\usepackage[utf8]{inputenc} % allow utf-8 input
\usepackage[T1]{fontenc}    % use 8-bit T1 fonts
\usepackage{hyperref}       % hyperlinks
\usepackage{amsfonts}       % blackboard math symbols
\usepackage{nicefrac}       % compact symbols for 1/2, etc.
\usepackage{microtype}      % microtypography
\usepackage{color}
\usepackage{algorithm}
\usepackage{algpseudocode}
\usepackage{mathtools}
\usepackage{enumitem}
\usepackage{epstopdf}
\usepackage{bm}
\usepackage{tabulary}
\usepackage[font=small]{caption}

\newtheorem{problem}{Problem}
\newtheorem{assumption}{Assumption}

\newcommand{\paren}[1]{\ensuremath{\left( #1\right)}}

\newcommand{\clint}[1]{\ensuremath{\left[ #1\right]}}
\renewcommand{\set}[1]{\ensuremath{\left\{ #1\right\}}}
\newcommand{\matr}[1]{\ensuremath{\clint{\begin{array} #1 \end{array}}}}
\newcommand{\norm}[1]{\ensuremath{\left\| #1\right\|}}
\newcommand{\snorm}[1]{\ensuremath{\| #1\|}}

\newcommand{\mbf}[1]{\ensuremath{\bm{#1}}}

\newcommand{\K}{\ensuremath{\mathcal{K}}}

\newcommand{\R}{\ensuremath{\mathbb{R}}}
\newcommand{\RR}{\ensuremath{\mbf{\mathfrak{R}}}}

\renewcommand{\O}{\ensuremath{\mathcal{O}}}
\newcommand{\C}{\ensuremath{\mathcal{C}}}
\newcommand{\T}{\ensuremath{\mathcal{T}}}

\newcommand{\Hinf}{\ensuremath{\mathcal{H}_{\infty}}}
\newcommand{\Ht}{\ensuremath{\mathcal{H}_{2}}}

\DeclareMathOperator{\Tr}{\mathrm{tr}}

\title[Sample Complexity of Kalman Filtering for Unknown Systems]{Sample Complexity of Kalman Filtering for Unknown Systems}
\usepackage{times}
\author{%
 \Name{Anastasios Tsiamis*} \Email{atsiamis@seas.upenn.edu}\\
 \Name{Nikolai Matni*} \Email{nmatni@seas.upenn.edu}\\
 \Name{George J. Pappas*} \Email{pappasg@seas.upenn.edu}\\
}

\begin{document}

\maketitle
\vspace{-0.5cm} 
\begin{abstract}%
In this paper, we consider the task of designing a Kalman Filter (KF) for an unknown and partially observed autonomous linear time invariant system driven by process and sensor noise.  To do so, we propose studying the following two step process: first, using system identification tools rooted in subspace methods, we obtain coarse finite-data estimates of the state-space parameters and Kalman gain describing the autonomous system; and second, we use these approximate parameters to design a filter which produces estimates of the system state.  We show that when the system identification step produces sufficiently accurate estimates, or when the underlying true KF is sufficiently robust, that a Certainty Equivalent (CE) KF, i.e., one designed using the estimated parameters directly, enjoys provable sub-optimality guarantees.  We further show that when these conditions fail, and in particular, when the CE KF is marginally stable (i.e., has eigenvalues very close to the unit circle), that imposing additional robustness constraints on the filter leads to similar sub-optimality guarantees.  We further show that with high probability, both the CE and robust filters have mean prediction error bounded by $\tilde O(1/\sqrt{N})$, where $N$ is the number of data points collected in the system identification step.  To the best of our knowledge, these are the first end-to-end sample complexity bounds for the Kalman Filtering of an unknown system.  
\end{abstract}


\section{Introduction}\label{Section_Introduction}
Time series prediction is a fundamental problem across control theory~\citep{kailath2000linear}, economics~\citep{bauer2002estimating} and machine learning. In the case of autonomous linear time invariant (LTI) systems driven by Gaussian process and sensor noise:
	\begin{equation}\label{EQN_System_General}
z_{k+1}=Az_k+w_k, \
y_k=Cz_k+v_k
\end{equation}
the celebrated Kalman Filter (KF) has been the standard method for prediction~\citep{anderson2005optimal} . When model~\eqref{EQN_System_General} is known, the KF minimizes the mean square prediction error.
However, in many practical cases of interest (e.g., tracking moving objects, stock price forecasting), the state-space parameters are not known and must be learned from time-series data.  This system identification step, based on a finite amount of data, inevitably introduces parametric errors in model~\eqref{EQN_System_General}, which leads to a KF with suboptimal prediction performance~\citep{el2001robust}.

In this paper, we study this scenario, and provide finite-data estimation guarantees for the Kalman Filtering of an unknown autonomous LTI system~\eqref{EQN_System_General}.
We consider a simple two step procedure.
In the first step, using system identification tools rooted in subspace methods, we obtain finite-data estimates of the state-space parameters, and Kalman gain describing system~\eqref{EQN_System_General}. Then, in the second step, we use these approximate parameters to design a filter which predicts the system state.  We provide an end-to-end analysis of this two-step procedure, and characterize the sub-optimality of the resulting filter in terms of the number of samples used during the system identification step, where the sub-optimality is measured in terms of the mean square prediction error of the filter.  A key insight that emerges from our analysis is that using a Certainty Equivalent (CE) Kalman Filter, i.e., using a KF computed directly from estimated parameters, can yield poor estimation performance if the resulting CE KF has eigenvalues close to the unit circle.  To address this issue, we propose a Robust Kalman Filter that mitigates these effects and that still enjoys provable sub-optimality guarantees.

Our main contributions are that:
i) we show that if the system identification step produces sufficiently accurate estimates, or if the underlying true KF is sufficiently robust, then the CE KF has near optimal mean square prediction error, 
ii)  we show when the CE KF is marginally stable, i.e., when it has eigenvalues close to the unit circle, that a Robust KF synthesized by explicitly imposing bounds on the magnitude of certain closed loop maps of the system enjoys similar mean square prediction error bounds as the CE KF, while demonstrating improved stability properties, and 
iii) we integrate the above results with the finite-data system identification guarantees of \cite{tsiamis2019finite}, to provide, to the best of our knowledge, the first end-to-end sample complexity bounds for the Kalman Filtering of an unknown system. In particular, we show that the mean square estimation error of both the Certainty Equivalent and Robust Kalman filter produced by the two step procedure described above is, with high probability, bounded by $\tilde O(1/\sqrt{N})$, where $N$ is the number of samples collected in the system identification step.

\textbf{Related work.}
A similar two step process was studied for the Linear Quadratic (LQ) control of an unknown system in~\cite{dean2017sample,mania2019certainty}.  While LQ optimal control and Kalman Filtering are known to be dual problems, this duality breaks down when the state-space parameters describing the system dynamics are not known.  In particular, the LQ optimal control problem assumes full state information, making the system identification step much simpler -- in particular, it reduces to a simple least-squares problem.  In contrast, in the KF setting, as only partial observations are available, the additional challenge of finding an appropriate system order and state-space realization must be addressed.
On the other hand, in the KF problem one can directly estimate the KF gain from data, which makes analyzing performance of the CE KF simpler than the performance of the CE LQ optimal controller~\citep{mania2019certainty}.

System identification of autonomous LTI systems~\eqref{EQN_System_General} is referred to as stochastic system identification~\citep{van2012subspace}.  Classical results consider the asymptotic consistency of stochastic subspace system identification, as in~\cite{deistler1995consistency,bauer1999consistency}, whereas contemporary results seek to provide finite data guarantees~\citep{tsiamis2019finite,lee2019non}.
Finite data guarantees for system identification of partially observed systems can also be found in~\cite{oymak2018non,simchowitz2019semi,sarkar2019finite}, but these results focus on learning the non-stochastic part of the system, assuming that a user specified input is used to persistently excite the dynamics.

Classical approaches to robust Kalman Filtering can be found in~~\cite{el2001robust,sayed2001framework,levy2012robust}, where parametric uncertainty is explicitly taken into account during the filter synthesis procedure.  Although similar in spirit to our robust KF procedure, these approaches assume fixed parametric uncertainty, and do not characterize the effects of parametric uncertainty on estimation performance, with this latter step being key in providing end-to-end sample complexity bounds.  We also note that although not directly comparable to our work, the filtering problem for an unknown LTI system was also recently studied in the adversarial noise setting in~\cite{hazan2018spectral}, where a spectral filtering technique is used to directly predict the output bypassing the system identification step. In the stochastic noise case, online-learning of the Kalman Filter was studied in~\cite{kozdoba2019line}, where the goal is to predict a scalar output. This is different from our paper, where the goal is to learn a state-space representation of the KF; our analysis holds for multi-output systems as well.

\textbf{Paper structure.} In Sec.~\ref{Section_Formulation}, we formulate the problem, and in Sec.~\ref{Section_CE} and~\ref{Section_Dynamic}, we derive performance guarantees for the proposed CE and Robust Kalman filters.  In Sec.~\ref{Section_Combined_Result}, we provide end-to-end sample complexity bounds for our two step procedure, and demonstrate the effectiveness of our pipeline with a numerical example in Sec.~\ref{Section_Simulations}.  We end with a discussion of future work in Sec.~\ref{Section_Conclusion}.   All proofs, missing details, and a summary of the system identification results from~\cite{tsiamis2019finite} can be found in the Appendix.

\textbf{Notation.} We let bold symbols denote the frequency representation of signals. For example, $\mbf{\Phi}=\sum_{t=0}^{\infty}\Phi_tz^{-t}$. If $M$ is stable  with spectral radius $\rho(M)<1$, then we denote its resolvent by $\RR_{M}\triangleq (zI-M)^{-1}$. The $\Ht$ system norm is defined by $\norm{\mbf \Phi}_{\Ht}^2\triangleq \sum_{t=0}^{\infty} \norm{\Phi_t}_F^2$, where $\norm{\cdot}_F$ is the Frobenius norm. The $\Hinf$ system norm is defined by $\norm{\mbf \Phi}_{\Hinf}\triangleq \sup_{\norm{z}=1} \norm{\mbf \Phi(z)}_2$, where $\norm{\cdot}_2$ is the spectral norm. Let $\frac{1}{z}\mathcal{RH}_{\infty}$ be the set of real rational stable strictly proper transfer matrices.

\section{Problem Formulation}\label{Section_Formulation}
For the remainder of the paper, we consider the Kalman Filter form of system~\eqref{EQN_System_General}:
\begin{equation}\label{EQN_System}
x_{k+1}=Ax_{k}+Ke_{k}, \ 
y_{k}=Cx_k+e_k,
\end{equation}
where $x_k\in \R^n$ is the prediction (state), $y_k\in\R^m$ is the output, and $e_k\in\R^{n}$ is the innovation process.
The innovations $e_k$ are assumed to be i.i.d. zero mean Gaussians, with positive definite covariance matrix $R$, and the initial state is assumed to be $x_0=0$. 
In general, the system~\eqref{EQN_System_General} driven by i.i.d. zero mean Gaussian process and sensor noise
is equivalent to system~\eqref{EQN_System} for a suitable gain matrix $K$, as both noise models produce outputs with identical statistical properties~\cite[Chapter 3]{van2012subspace}. 
We make the following assumption throughout the rest of the paper.
\begin{assumption}\label{ASS_Kalman}
	Matrices $A,C,K,R$ are unknown, and the pair $(A,C)$ is observable.  Both the matrices $A$ and $A-KC$ have spectral radius less than $1$, i.e.,  $\rho(A)<1$ and $\rho(A-KC)<1$. 
\end{assumption}
The observability assumption is standard, and the stability of $A-KC$  follows from the properties of the Kalman filter~\citep{anderson2005optimal}.  We note that the filter synthesis procedures we propose can be applied even if $\rho(A)\geq 1$ -- however, in this case, we are unable to guarantee bounded estimation error for the resulting CE and robust KFs (see Theorem~\ref{THM_CE_performance} and Lemma~\ref{LEM_error_analysis}).  

Our goal is to provide end-to-end sample complexity bounds for the two step pipeline illustrated in Fig.~\ref{Figure_CoarseID_architecture}. First, we collect a trajectory $\{y_t\}_{t=0}^N$ of length $N$ from system \eqref{EQN_System}, and use system identification tools with finite data guarantees  to learn the parameters $\hat{A},\hat{C},\hat{K},\hat{R}$ and bound the corresponding parameter uncertainties by  $(\epsilon_A,\epsilon_C,\epsilon_K,\epsilon_R)$. Second, we use these approximate parameters to synthesize a filter from the following class:
	\begin{equation}\label{EQN_Def_Predictor_class}
	\tilde{x}_{k}=\hat{A}\tilde{x}_{k-1}+\sum_{t=1}^{k}L_t (y_{k-t}-\hat{C}\tilde{x}_{k-t}),\quad 	\tilde{J}\triangleq \sqrt{\lim_{T\rightarrow \infty}\frac{1}{T}\sum_{k=0}^{T}\snorm{\tilde{x}_k-x_k}_2^2}
	\end{equation}
	where $\set{L_{t}}_{t=1}^{\infty}$ are to be designed and $\tilde{J}$ is the filter's mean square prediction error as defined with respect to the optimal KF.
Note that the predictor class above includes the CE KF -- see Section~\ref{Section_CE} -- and that if the the true system parameters are known, i.e., if $\hat{A}=A$, $\hat{C}=C$, $\hat{K}=K$, then the optimal mean squared prediction error $\tilde{J}=0$ is achieved.

\begin{figure}[t]
	\centering
	\includegraphics[scale=1]{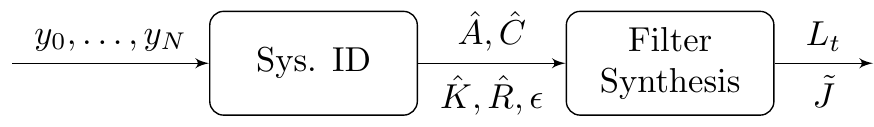}
	\caption{The proposed identification and filter synthesis pipeline.  Using a single trajectory of $N$ samples $\{y_t\}_{t=0}^N$ generated by system~\eqref{EQN_System}, a system identification algorithm computes estimates of $(\hat{A},\hat{C},\hat{K},\hat{R})$ with corresponding identification error bounds $\epsilon:=\max(\epsilon_A,\epsilon_C,\epsilon_K,\epsilon_R)$. Then, using these estimates, we synthesize a filter defined by dynamic gains $\set{L_{t}}_{t=1}^{\infty}$, which has mean square prediction error $\tilde{J}$, defined in~\eqref{EQN_Def_Predictor_class}.}\label{Figure_CoarseID_architecture}
	\vspace{-0.5cm}
\end{figure}
\begin{problem}[End-to-end Sample Complexity]\label{Prob_Sample_Complexity}
Fix a failure probability $\delta > 0$.  Given a single trajectory $y_0,\dots,y_N$ of system~\eqref{EQN_System}, compute system parameter estimates
$\hat{A},\hat{C},\hat{K},\hat{R}$, and design a Kalman filter in class~\eqref{EQN_Def_Predictor_class}, defined by gains $\set{L_{t}}_{t=1}^{\infty}$, such that with probability at least $1-\delta$, we have that $
\tilde{J}\le \epsilon_{J}$, so long as $N \ge \mathrm{poly}(1/\epsilon_J,\log(1/\delta))$.
\end{problem}

To address Problem \ref{Prob_Sample_Complexity}, we will: i) leverage recent results regarding the the sample complexity of stochastic system identification,  ii) provide estimation guarantees for certainty equivalent as well as robust Kalman filter designed using the identified system parameters (see Problem~\ref{Prob_Synthesis} below), and (iii) provide end-to-end performance guarantees by integrating steps (i) and (ii) (see  Problem~\ref{Prob_Sample_Complexity} above).

Recently ~\cite{tsiamis2019finite} provided a finite sample analysis for stochastic system identification which provides bounds on the identification error $\epsilon:=\max(\epsilon_A,\epsilon_C,\epsilon_K,\epsilon_R)$.  Leveraging these results, we focus next on solving the Filter Synthesis task described below using both a certainty equivalent Kalman filter as well as a robust Kalman filter.
\begin{problem}[Near Optimal Kalman Filtering of an Uncertain System]\label{Prob_Synthesis}
	Consider system~\eqref{EQN_System}. Let $\hat{A},\hat{C},\hat{K},\hat{R}$ be estimates satisfying\footnote{
	 In practice, estimating the parameters of a partially observed system~\eqref{EQN_System} is ill-posed, in that any similarity transformation $S$ can be applied to generate parameters $(S^{-1}AS, CS,S^{-1}K, R)$ describing the same system, and the bounds described hold for \emph{some} similarity transformation $S$.  All results in this paper apply nearly as is to the general case of $S\neq I$ under suitable assumptions -- more details can be found in the extended version.}
$
	\snorm{A-\hat{A}}_2\le \epsilon_A,\,\snorm{C-\hat{C}}_2\le \epsilon_C,\,\snorm{K-\hat{K}}_2\le \epsilon_K,\,\snorm{R-\hat{R}}_2\le \epsilon_R.
$
	Design a Kalman filter in class~\eqref{EQN_Def_Predictor_class}, defined by gains $\{L_t\}_{t=0}^\infty$, with mean square prediction error decaying with the size of the parameter uncertainty, i.e., such that
$
\tilde{J}\le O(\epsilon_A,\epsilon_C,\epsilon_K,\epsilon_R)$.
\end{problem}

%!TEX root = main.tex
\section{Estimation Guarantees for Certainty Equivalent Kalman Filtering}\label{Section_CE}
For the certainty equivalent Kalman filter, we directly use the estimated state-space parameters from the system identification step.  
Based on the estimated $\hat{K},\hat{R}$ we compute the covariance:
\[
\matr{{cc}\hat{Q}&\hat{S}\\\hat{S}^*&\hat{R}}\triangleq \mathbb{E}\matr{{c}\hat{K}e_k\\e_k}\matr{{cc}e^*_k\hat{K}^*&e^*_k}=\matr{{c}\hat{K}\hat{R}^{1/2}\\ \hat{R}^{1/2}}\matr{{cc}\hat{R}^{1/2}\hat{K}^*& \hat{R}^{1/2}}.
\]
Then, based on standard Kalman filter theory, we compute the stabilizing solution\footnote{A stabilizing solution $P$ to the Riccati equation defines a Kalman gain $L_{CE}$ such that $\rho(\hat{A}-L_{CE}\hat{C})<1$.} of the following Riccati equation with correlation terms~\citep{kailath2000linear}:
\begin{equation}\label{EQN_DEF_CE_Riccati}
P=\hat{A}P\hat{A}^*+\hat{Q}-(\hat{A}P\hat{C}^*+\hat{S})(\hat{C}P\hat{C}^*+\hat{R})^{-1}(\hat{C}P\hat{A}^*+\hat{S}^*).
\end{equation}
Then, the  CE Kalman filter gain is static and takes the form
\begin{equation}\label{EQN_DEF_CE_predictor}
L_1=L_{CE}\triangleq (\hat{A}P\hat{C}^*+\hat{S})(\hat{C}P\hat{C}^*+\hat{R})^{-1},\,L_t=0, \text{ for }t=2,\dots.
\end{equation}
Trivially, if $\rho(\hat{A}-\hat{K}\hat{C})<1$, then the stabilizing solution of the Riccati equation is $P=0$ with $L_{CE}=\hat{K}$; the solution does not depend on $\hat{R}$. The next result shows that if the underlying true Kalman filter is sufficiently robust, as measured by a spectral decay rate, and that estimation parameter errors are sufficiently small, then the CE Kalman filter achieves near optimal performance. 
\begin{theorem}[Near Optimal Certainty Equivalent Kalman Filtering]\label{THM_CE_performance}
	Consider Problem~\ref{Prob_Synthesis} and the CE KF~\eqref{EQN_DEF_CE_predictor}. For any $\rho(A-KC)\leq\rho<1$, define 
	$
	\tau(A-KC,\rho)\triangleq \sup_{t\ge 0}{\norm{\paren{A-KC}^{t}}_2\rho^{-t}}.
	$ 
	If the robustness condition 
	$2	\tau(A-KC,\rho)\cdot\left(\epsilon_A+\epsilon_C(\norm{K}_2+\epsilon_K)+
	\epsilon_K\norm{C}_2\right)\le 1-\rho$ %\frac{1-\rho}{2	\tau(A-KC,\rho)} \]
	is satisfied,   then $L_{CE}=\hat{K}$ and:
	\[
	\tilde{J}\le \sqrt{3}\bar{\C}\epsilon\norm{\matr{{c} \RR_AK\\ I}R^{1/2}}_{\Ht}
	\]	
	where $\epsilon=\max\set{\epsilon_A,\epsilon_C,\epsilon_K}$, $
	\bar{\C}=2\frac{\tau(A-KC,\rho)}{1-\rho}\paren{1+\norm{K}_2+\epsilon_K}$ and $\RR_{A}=(zI-A)^{-1}$. 
\end{theorem}

The transient behavior of the CE Kalman filter is governed by the closed loop eigenvalues of $\hat{A}-\hat{K}\hat{C}$, with performance degrading as eigenvalues approach the unit circle.  This may occur if the estimation errors $(\epsilon_A,\epsilon_C,\epsilon_K)$ are large enough to cause  $\rho(\hat{A}-\hat{K}\hat{C})\approx 1$ even if the true system has spectral radius $\rho(A-KC)<1$. We show in the next section that this undesirable scenario can be avoided by explicitly constraining the transient response of the resulting Kalman filter to satisfy certain robustness constraints.

\section{Estimation Guarantees for Robust Kalman Filtering}\label{Section_Dynamic}
To address the possible poor performance of the CE Kalman filter when model uncertainty is large, we propose to search over  dynamic filters~\eqref{EQN_Def_Predictor_class} subject to additional robustness constraints on their transient response. 
Using the System Level Synthesis (SLS) framework~\citep{wang2019system,anderson2019system} for Kalman Filtering~\citep{wang2015localized}, we
parameterize the class of dynamic filters~\eqref{EQN_Def_Predictor_class} subject to additional robustness constraints in a way that leads to convex optimization problems. 

 For a given dynamic predictor $\mbf L(z) = \sum_{t=0}^\infty z^{-t}L_{t+1}$, we define the closed loop \emph{system responses}:
\begin{equation}\label{EQN_Def_responses}
\mbf \Phi_w(z)\triangleq (zI-\hat{A}+\mbf L \hat{C})^{-1},\, \mbf \Phi_v(z)\triangleq -(zI-\hat{A}+\mbf L \hat{C})^{-1}\mbf L.
\end{equation}
In \citep{wang2015localized}, it is shown that these responses are in fact the closed loop maps from process and sensor noise $(\mbf w,\mbf v)$ to state estimation error, and that the filter gain achieving the desired behavior can be recovered via $\mbf L=-\mbf \Phi^{-1}_w\mbf \Phi_v$ so long as the responses $(\mbf \Phi_w,\mbf \Phi_v)$ are constrained to lie in an affine space defined by the system dynamics.  By expressing the mean squared prediction error of the filters~\eqref{EQN_Def_Predictor_class} in terms of their system responses, we are able to clearly delineate the effects of parametric uncertainty from the cost of deviating from the CE Kalman filter.
\begin{lemma}[Error analysis]\label{LEM_error_analysis}
	Consider system~\eqref{EQN_System}. Let $\Delta_A\triangleq A-\hat{A}$, $\Delta_C\triangleq C-\hat{C}$, $\Delta_K\triangleq K-\hat{K}$. Any filter~\eqref{EQN_Def_Predictor_class} with parameterization~\eqref{EQN_Def_responses} has mean squared prediction error given by
	\[
	\tilde{J}=\norm{\matr{{cc}\mbf\Phi_w & \mbf\Phi_v}\set{\matr{{cc}\Delta_A&\Delta_K\\  \Delta_C&0}\matr{{c}\mbf \RR_AK\\ I}+\matr{{c}\hat{K}\\ I}}R^{1/2}}_{\Ht}
	\]
\end{lemma}
Based on the previous lemma, we can upper bound the mean squared prediction error of filters~\eqref{EQN_Def_Predictor_class} by
\begin{align*}
\tilde{J}\le& \underbrace{\sqrt{3}\epsilon\norm{\matr{{cc}\mbf\Phi_w & \mbf\Phi_v}}_{\Ht}\norm{\matr{{c}\mbf \RR_AK\\ I}R^{1/2}}_{\Hinf}}_{\text{parameter uncertainty term}}+\underbrace{\norm{\mbf \Phi_w \hat{K}+\mbf \Phi_v}_{\Ht}\norm{R^{1/2}}_2}_{\text{suboptimality term}},
\end{align*}
where $\epsilon=\max\set{\epsilon_A,\epsilon_C,\epsilon_K}$. This upper bound clearly separates the effects of parameter uncertainty, as captured by the first term, and the performance cost incurred by the filter $\mbf L$ due to its deviation from the CE Kalman gain $\hat{K}$, as captured by the second.   In order to optimally tradeoff between these two terms, we propose the following robust SLS optimization problem:
\begin{equation}\label{EQN_Robust_Optimization}
\begin{aligned}
&\min_{\mbf\Phi_w,\mbf\Phi_v}\norm{\mbf\Phi_w\hat{K}+\mbf\Phi_v}_{\mathcal{H}_2}
\\ &\text{s.t.} \, \norm{\matr{{cc}\mbf\Phi_w\ \mbf\Phi_v}}_{\Ht}\le \C\\
& \mbf\Phi_w(zI-\hat{A})-\mbf \Phi_v\hat{C}=I,\,\mbf \Phi_w,\,\mbf \Phi_v\in \frac{1}{z}\mathcal{RH}_{\infty}\\
\end{aligned}
\end{equation}
where the constant $\C$ is a regularization parameter, and the affine constraint $\mbf\Phi_w(zI-\hat{A})-\mbf \Phi_v\hat{C}=I,\,\mbf \Phi_w,\,\mbf \Phi_v\in \frac{1}{z}\mathcal{RH}_{\infty}$
parameterizes all filters of the form~\eqref{EQN_Def_Predictor_class} that have bounded mean squared prediction error (see~\cite{wang2015localized} for more details). % Note that by constraining the $\Ht$ norm of the system responses, we restrict the effect of the parameter uncertainty term, while minimizing the suboptimality term.
%We could also include $\hat{R}$ in the optimization problem, but we omit it for simplicity. 
%The intuition behind the optimization problem is the following. The objective measures how much we deviate from the CE predictor. 
%Notice that we impose the additional robustness constraint $\norm{\matr{{cc}\mbf\Phi_w&\mbf\Phi_v}}_{\Hinf}\le \C$. 
%The constant $\C$ is a design choice and controls the tradeoff between robustness and certainty equivalence. Selecting a very large $\C$ leads to predictors with similar performance to CE. On the other hand, a small $\C$ imposes better stability margins but leads to potentially suboptimal performance. The intuition behind the constraint is explained in~Theorem~\ref{THM_SLS_Guarantees}; it guarantees that the parameter uncertainty will have a bounded effect on the prediction error.
%The regularization parameter C can vary between an upper bound, the $\Ht$ norm of the CE KF, and a lower bound, the minimum $\Ht$ filter based on $\hat A$ and $\hat C$. This corresponds to the case where we only care about robustness. Any lower value for C will result in an infeasible problem. 
%Both of these bounds can be computed based only on the data-driven \hat A, \hat C, L_{CE}. If we want more robustness we choose a C close to the lower bound; if we want less robustness we choose a larger C close to the upper bound. 
As we formalize in the following theorem, for appropriately selected regularization parameter $\C$ and sufficiently accurate estimation errors $(\epsilon_A,\epsilon_C)$, the robust KF has near optimal mean square estimation error.
\begin{theorem}[Robust Kalman Filter]\label{THM_SLS_Guarantees}
	Consider Problem~\ref{Prob_Synthesis} with Kalman filters from class~\eqref{EQN_Def_Predictor_class} synthesized using the robust SLS optimization problem \eqref{EQN_Robust_Optimization}. 
%
%		\item[a)] 	Let $\mbf L=-\mbf \Phi^{-1}_w\mbf \Phi_v$, where $(\mbf \Phi_w, \mbf \Phi_v)$ is the solution to optimization problem~\eqref{EQN_Robust_Optimization}. Then, the prediction error is upper bounded by:
%		\begin{equation}\label{EQN_Robust_Upper_Bound}
%		\tilde{J} \le \sqrt{3}\C\epsilon\norm{\matr{{c}\mbf \RR_AK\\ I}}_{\Hinf}\snorm{R^{1/2}}_2+\text{\normalfont opt}(\C)\snorm{R^{1/2}}_2
%		\end{equation}
%		where $\text{\normalfont opt}(\C)$ is the optimal value of~\eqref{EQN_Robust_Optimization} and is upper bounded: $\text{\normalfont opt}(\C)\le \C \paren{\snorm{K}_2+1+\epsilon}$.
%
%		
If the regularization parameter is chosen such that $\C\ge2 (1+\norm{K}_2)\norm{\RR_{A-KC}}_{\Ht}$, and further, the estimation errors $(\epsilon_A,\epsilon_C)$ are such that
		\begin{equation}\label{EQN_Robustness_Condition}
		(\epsilon_A+\epsilon_C\norm{K}_2)\norm{\RR_{A-KC}}_{\Hinf}\le 1/2
		\end{equation}
		then the robust SLS optimization problem is feasible, and the synthesized robust Kalman filter has mean squared prediction error upper-bounded by
		\begin{equation}\label{EQN_CERobust_Bound}
		\tilde{J} \le \sqrt{3}\C\epsilon\norm{\matr{{c}\mbf \RR_AK\\ I}}_{\Hinf}\snorm{R^{1/2}}_2+2\epsilon\norm{\RR_{A-KC}}_{\Ht}\snorm{R^{1/2}}_2,
		\end{equation}
		%%%previous statement
		%\C2\sqrt{2}\norm{\RR_{A}}_{\Ht}\sqrt{\norm{R}_2}\epsilon+\C\epsilon_K\sqrt{\norm{R}_2}+2\epsilon_K\norm{\RR_{A-KC}}_{\Ht}\sqrt{\norm{\hat{R}}_2} 
where $\epsilon=\max\set{\epsilon_A,\epsilon_C,\epsilon_K}$.
\end{theorem}
We further note that whenever the system responses induced by the CE Kalman filter $\mbf{\tilde\Phi}_w \triangleq (zI-\hat A+\hat{K}\hat{C})^{-1}$, $\mbf{\tilde\Phi}_v \triangleq-(zI-\hat A+\hat{K}\hat{C})^{-1}\hat{K}$ are a feasible solution to optimization problem~\eqref{EQN_Robust_Optimization},  they are also optimal, resulting in a filter $\mbf L=\hat{K}$ with performance identical to the CE setting.

\section{End-to-End Sample Complexity for the Kalman Filter }\label{Section_Combined_Result}

Theorems~\ref{THM_CE_performance} and~\ref{THM_SLS_Guarantees} provide two different solutions to Problem~\ref{Prob_Synthesis}.
Combining these theorems
with the finite data system identification guarantees of~\cite{tsiamis2019finite}, we now derive, to the best of our knowledge, the first end-to-end sample complexity bounds for the Kalman filtering of an unknown system. For both the CE and robust Kalman filter, we show that the mean squared estimation error defined in \eqref{EQN_Def_Predictor_class} decreases with rate $O(1/\sqrt{N})$ up to logarithmic terms, where $N$ is the number of samples collected during the system identification step. 
The formal statement of the following theorem which addresses Problem~\ref{Prob_Sample_Complexity} can be found in~Theorem~\ref{THM_EndToEnd_Formal}.

\begin{theorem}[End-to-end guarantees, informal]\label{THM_EndToEnd}
Fix a failure probability $\delta \in (0,1)$, and assume that we are given a sample trajectory $\{y_t\}_{t=0}^N$ generated by system~\eqref{EQN_System}.  Then as long as $N\ge \mathrm{poly}(\log(1/\delta))$, we have with probability at least $1-\delta$ that the identification and filter synthesis pipeline of Fig.~\ref{Figure_CoarseID_architecture}, with system identification performed as in \cite{tsiamis2019finite} and filter synthesis performed as in Sections~\ref{Section_CE},~\ref{Section_Dynamic}, achieves mean squared prediction error satisfying
\[
\tilde{J}\le \C_{ID}\C_{KF}\tilde{O}\left(\sqrt{\frac{\log(1/\delta)}{N}}\right),\text{ where }\C_{KF}=\frac{1}{1-\rho(A-KC)}\frac{1}{1-\rho(A)}
\]	
and $\C_{ID}$ captures the difficulty of identifying system~\eqref{EQN_System} (see~\eqref{EQN_Cid} in the Appendix).  Here, $\tilde{O}$ hides constants, other system parameters, and logarithmic terms.
\end{theorem}

The bound derived in Theorem~\ref{THM_EndToEnd} highlights an interesting tension between how easy it is to identify the unknown system, and the robustness of the underlying optimal Kalman filter.  The constant $\C_{KF}$ captures how robust the underlying open loop system $A$ and closed loop Kalman filter $A-KC$ are, as measured by their spectral gaps $1-\rho(A)$ and $1-\rho(A-KC)$.  In particular, we expect $\C_{KF}$ to be small for systems that admit optimal KFs with favorable robustness and transient performance.  In contrast, the constant $\C_{ID}$ captures how easy it is to identify a system: recent results for the fully observed setting \citep{simchowitz2018learning,sarkar2018fast} suggest that systems with \emph{larger} spectral radius are in fact easier to identify, as they provide more ``signal'' to the identification algorithm.  In this way, our upper bound suggests that systems which properly balance between these two properties, robust transient performance and ease of identification, enjoy favorable sample complexity.

We also note that the degradation of our bound with the inverse of the spectral gap $1-\rho(A)$ appears to be a limitation of the proposed offline two step architecture -- indeed, Lemma~\ref{LEM_error_analysis} suggests that any estimation error in the state-space parameters $(A,C)$ causes an increase in mean squared prediction error as $\snorm{\RR_A} \propto (1-\rho(A))^{-1}$ increases. 
It remains open as to whether other prediction architectures would suffer from the same limitation.

\section{Simulations}\label{Section_Simulations}
We perform Monte Carlo simulations of the proposed pipeline for the system 
 \[
A=\matr{{ccc}0.8&1&0\\0&0.9&1\\0&0&0.9},\, C=\matr{{ccc}1&0&0}, \,K=\matr{{ccc}1.5320&	0.9401&
	0.1923}^*,\,R=10.6414.
\]
for varying sample lengths $N$. We simulate both the CE and robust Kalman filters, and set the regularization parameter to $\C=10$ in the robust SLS optimization problem~\eqref{EQN_Robust_Optimization}.
For each iteration, we first simulate system~\eqref{EQN_System} to obtain $N$ output samples. Then, we perform system identification to obtain the system parameters, after which we synthesize both CE and robust Kalman filters. Finally, we compute the mean prediction error of the designed filters. 

For the identification scheme, we used the variation of the MOESP algorithm~\cite{qin2006overview}, which is more sample efficient in practice than the one analyzed in~\cite{tsiamis2019finite}--see Algorithm~\ref{ALG_identification} and Section~\ref{Section_SVD}. 
The basis of the state-space representation returned by the subspace algorithm is data-dependent and varies with each simulation. For this reason, to compare the performance across different simulations, we compute the mean square error in terms of the original state space basis.
Note that the SLS optimization problem~\eqref{EQN_Robust_Optimization} is semi-infinite since we optimize over the infinite variables $\set{\Phi_{w,t}}_{t=1}^{\infty}$ and $\set{\Phi_{v,t}}_{t=0}^{\infty}$. To deal with this issue, we optimize over a finite horizon $T$--see for example~\cite{dean2018regret}, which makes the problem finite and tractable. Here, we selected $T=30$.

\begin{figure}[t] \centering{
		\subfigure[\footnotesize CE Kalman Filter]{
			\centering
			\includegraphics[scale=0.44]{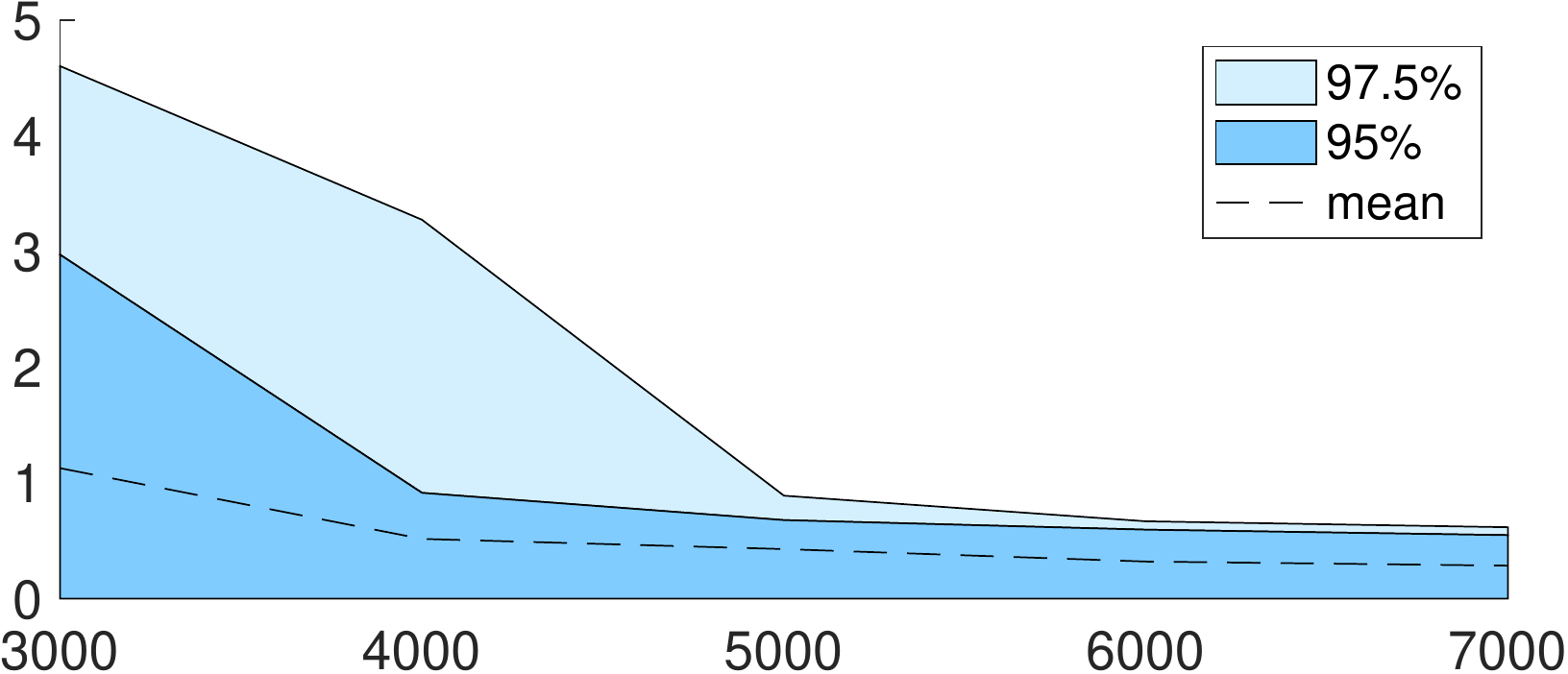}
	}
	\subfigure[\footnotesize Robust Kalman Filter]{
	\centering
	\includegraphics[scale=0.44]{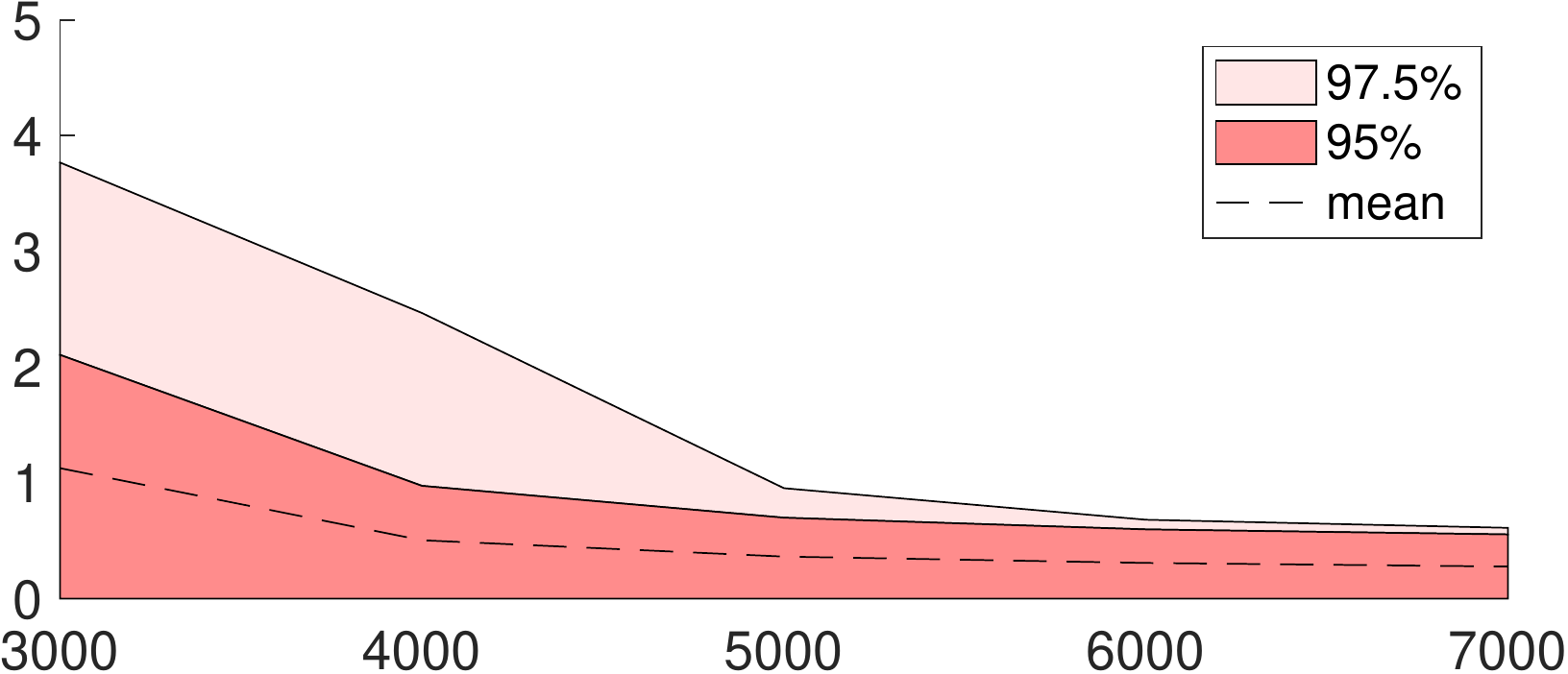}
}
}
	\caption{The $95\%$ and $97.5\%$ empirical percentiles for the mean squared prediction error $\tilde{J}$ of the CE and Robust Kalman filters. We run 1000 Monte Carlo simulations for different sample lengths $N$ ($x$-axis, number of samples).}
	\label{Figure_error_decay}
\end{figure}

Figure~\ref{Figure_error_decay} (a) and (b) show the empirically computed mean squared prediction errors of the CE and Robust Kalman filters, with the mean, 95th, and 97.5th percentiles being shown. Notice that both errors decrease with a rate of $1/\sqrt{N}$, and that while the average behavior of both filters is quite similar, there is a noticeable gap in their tail behaviors.  We observe that the most significant gap between the CE and Robust Kalman filters occurs when the eigenvalues of the CE matrix $\hat{A}-L_{CE}\hat{C}$ are close to the unit circle. Fig.~\ref{Figure_performance_improvement} shows the empirical distribution of mean squared prediction errors conditioned on the event that $\rho(\hat{A}-L_{CE}\hat{C})>0.97$. In this case, the CE filter can exhibit \emph{extremely} poor mean squared prediction error, with the worst observed error (not shown in Fig.~\ref{Figure_performance_improvement} in the interst of space) approximately equal to 70 -- in contrast, the worst error exhibited by the robust Kalman filter was approximately equal to 5. Thus, we were able to achieve a 14x reduction in worst-case mean squared error.  For some simulations the robust KF can exhibit worse performance compared to the CE Kalman filter. However, over all simulations, the mean squared error achieved by the robust Kalman filter was at most 1.64x greater than that achieved by CE Kalman filter.

\begin{figure}[t] \centering{
		\includegraphics[scale=0.44]{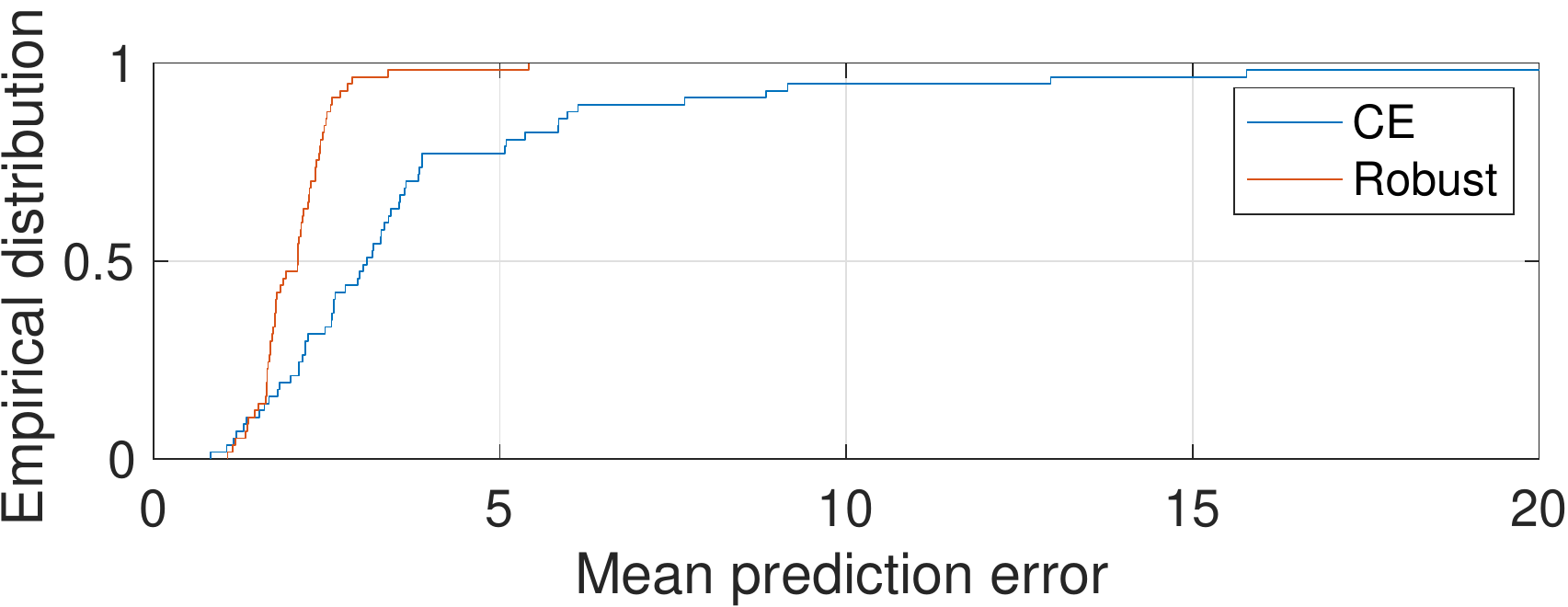}}   
	\caption{Performance improvement for the robust KF {conditioned on the event} that $\rho(\hat{A}-L_{CE}\hat{C})>0.97$. }
	\label{Figure_performance_improvement}
\end{figure}

\section{Conclusions \& Future work}
\label{Section_Conclusion}
In this paper, we proposed and analyzed a system identification and filter synthesis pipeline.  Leveraging contemporary finite data guarantees from system identification \citep{tsiamis2019finite}, as well as novel parameterizations of robust Kalman filters \citep{wang2015localized}, we provided, to the best of our knowledge, the first end-to-end sample complexity bounds for the Kalman filtering of an unknown autonomous LTI system.  Our analysis revealed that, depending on the spectral properties of the CE Kalman filter, a robust Kalman filter approach may lead to improved performance.  In future work, we would like to explore how to improve robustness and performance  by further exploiting information about system uncertainty, as well as how to integrate our results into an optimal control framework, such as Linear Quadratic Gaussian control.

\bibliography{Literature_Identification}
\appendix
\section{Properties of the CE Kalman Filter}
The following result, which follows from the theory of non-stabilizable Riccati equations~\cite{chan1984convergence}, describes the form of the certainty equivalent gain.
\begin{lemma}\label{LEM_CE_Form}
	Consider the assumptions of Problem~\ref{Prob_Synthesis}. Assume that $(\hat{A},\hat{C})$ is observable and $\hat{R}$ is positive definite. The CE Kalman filter gain $L_{CE}$~\eqref{EQN_DEF_CE_predictor} has the following properties:
	\begin{itemize}[wide, labelwidth=!, labelindent=0pt]
		\item If $\rho(\hat{A}-\hat{K}\hat{C})<1$, then $L_{CE}=\hat{K}$ and $\hat{A}-L_{CE}\hat{C}$ is asymptotically stable.
		\item If $\rho(\hat{A}-\hat{K}\hat{C})>1$, and $\hat{A}-\hat{K}\hat{C}$ has no eigenvalues on the unit circle, then $\hat{A}-L_{CE}\hat{C}$ is asymptotically stable. 
		\item If $\hat{A}-\hat{K}\hat{C}$ has eigenvalues on the unit circle, then~\eqref{EQN_DEF_CE_Riccati}  does not admit a stabilizing solution.
	\end{itemize}
\end{lemma}
\begin{proof}
After some algebraic manipulations--see also~\cite{kailath2000linear}, the Riccati equation~\eqref{EQN_DEF_CE_Riccati} can be rewritten as:
\begin{align*}
P=(\hat{A}-\hat{K}\hat{C})P(\hat{A}-\hat{K}\hat{C})^*-(\hat{A}-\hat{K}\hat{C})P\hat{C}^*(\hat{C}P\hat{C}^*+\hat{R})^{-1}\hat{C}P(\hat{A}-\hat{K}\hat{C})^*
\end{align*}
Notice that there is no $Q$ term in the equivalent algebraic Riccati equation.
If $\hat{A}-\hat{K}\hat{C}$ is already stable then the trivial solution $P=0$ is the stabilizing one. 
If $\hat{A}-\hat{K}\hat{C}$ is not asymptotically stable the results follow from Theorem~3.1 of~\cite{chan1984convergence}.
\end{proof}
\section{SLS preliminaries}
For this subsection, we assume that $\hat{A}=A,\hat{C}=C,\hat{K}=K,\hat{R}=R$. 
Using bold symbols to denote the frequency representation of signals, we can rewrite the original system equation~\eqref{EQN_System} and the predictor equation~\eqref{EQN_Def_Predictor_class} as:
\[
(zI-A+KC)\mbf{x}=K\mbf{y} ,\quad (zI-A+\mbf L C)\mbf{\tilde{x}}=\mbf L \mbf y.
\] 
Subtracting the two equations and using the fact that $\mbf y=C\mbf x +\mbf e $, we obtain:
\begin{equation*}
\mbf{x}-\mbf{\tilde x}=(zI-A+\mbf{L}C)^{-1}K\mbf e-(zI-A+\mbf{L}C)^{-1}\mbf{L}\mbf e
\end{equation*}
Define the responses to $K\mbf e$ and $\mbf e$ by $\mbf \Phi_{w}\triangleq(zI-A+\mbf L C)^{-1}$ and $\mbf \Phi_{v}\triangleq -(zI-A+\mbf L C)^{-1}\mbf L$ respectively. Then the error obtains the linear representation:
\[
\mbf{x}-\mbf{\tilde x}=(\mbf \Phi_w K+\mbf \Phi_v)\mbf e
\]
The case of $A\neq \hat A$, $C\neq \hat C$, $K\neq \hat K$ can be found in Lemma~\ref{LEM_error_analysis}.
The following result from~\cite{wang2015localized} parameterizes the set of stable closed-loop transfer matrices $\mbf L$.
\begin{proposition}[Predictor parameterization]\label{PROP_Parameterization}
	Consider system~\eqref{EQN_System}.	Let $\frac{1}{z}\mathcal{RH}_{\infty}$ denote the set of real rational stable strictly proper transfer matrices. The closed-loop responses $\mbf \Phi_w,\,\mbf \Phi_v$ from $K\mbf e$ and $\mbf e$ to $\mbf{x}-\mbf{\tilde x}$ can be induced by an internally stable predictor $\mbf L$ if and only if they belong to the following affine subspace:
	\begin{equation}\label{EQN_Affine_Constraint}
	\matr{{cc}\mbf \Phi_w& \mbf \Phi_v}\matr{{c}zI-A\\-C}=I,\,\mbf \Phi_w,\,\mbf \Phi_v\in \frac{1}{z}\mathcal{RH}_{\infty}.
	\end{equation}
	Given the responses, we can parameterize the prediction gain as $\mbf L=-\mbf \Phi^{-1}_w\mbf \Phi_v$. 
\end{proposition}
Let $\mbf \Phi_w=\sum_{t=0}^{\infty}\Phi_{w,t}z^{-t}$ and $\mbf \Phi_v=\sum_{t=0}^{\infty}\Phi_{v,t}z^{-t}$. The strictly proper condition enforces the constraint
$
\Phi_{w,0}=0,\Phi_{v,0}=0.
$
The affine constraints simply imply that the system responses $\mbf \Phi_w,\mbf \Phi_v$ should satisfy the linear system recursions:
\[
\Phi_{w,t+1}=\Phi_{w,t}A+\Phi_{v,t}C, \,t\ge 1,\quad \Phi_{w,1}=I
\]

Assuming that the predictor is internally stable, then the mean square error is equal to 
\[
\tilde{J}=\snorm{(\mbf \Phi_w K+\mbf \Phi_v)R^{1/2}}_{\Ht},
\]
where $\snorm{\cdot}_{\Ht}$ is the $\Ht$ system norm.
Hence, the error-free Kalman filter synthesis problem could be re-written as:
\[
\min_{\mbf \Phi_w, \mbf\Phi_v} \snorm{(\mbf \Phi_w K+\mbf \Phi_v)R^{1/2}}_{\Ht},\quad \text{s.t. } \eqref{EQN_Affine_Constraint}
\]
Of course, when the model knowledge is perfect, the solution to this problem is trivially $\mbf L=K$, $\mbf \Phi_w=(zI-A+KC)^{-1}$, $\mbf \Phi_v=-(zI-A+KC)^{-1}K$, $\tilde{J}=0$. 
\section{Proofs}

\subsection*{Proof of Theorem~\ref*{THM_CE_performance}}
	Let $\Delta_{A_{cl}}=(A-KC)-(\hat{A}-\hat{K}\hat{C})$.  
	By adding and subtracting $\hat{K}C$, we obtain the bound:
\[
	\snorm{\Delta_{A_{cl}}}\le \epsilon_A+\snorm{\hat{K}}_2\epsilon_C+\epsilon_K\snorm{C}_2\le \epsilon_A+(\norm{K}_2+\epsilon_K)\epsilon_C+
	\epsilon_K\norm{C}_2
\]
	Hence, from the robustness condition of the theorem it follows that
		\begin{equation}\label{APP_CE_EQN_AUX_1}
2\tau(A-KC,\rho)\snorm{\Delta_{A_{cl}}}_2\le 1-\rho
	\end{equation}

	Now, from Lemma~5 in~\cite{mania2019certainty} it follows that:
	\begin{equation}	\label{APP_CE_EQN_AUX_2}
	\snorm{(\hat{A}-\hat{K}\hat{C})^k}_2=\snorm{(A-KC-\Delta_{A_{cl}})^k}_2\le 	\tau(A-KC,\rho)\paren{	\tau(A-KC,\rho)\norm{\Delta_{A_{cl}}}_2+\rho}^k
	\end{equation}
	Combining~\eqref{APP_CE_EQN_AUX_1},~\eqref{APP_CE_EQN_AUX_2}, we finally obtain:
	\[
		\snorm{(\hat{A}-\hat{K}\hat{C})^k}_2\le \tau(A-KC,\rho) \paren{\frac{1+\rho}{2}}^k.
	\]
	Thus, the $\Hinf$ norm of $\RR_{\hat{A}-\hat{K}\hat{C}}$ is upper bounded by 
	\begin{align*}
	\norm{\RR_{\hat{A}-\hat{K}\hat{C}}}_{\Hinf}&\le \sum_{t=0}^{\infty}\snorm{(\hat{A}-\hat{K}\hat{C})^{t}}_2\\
	&\le \tau(A-KC,\rho) \sum_{k=0}^{\infty}  \paren{\frac{1+\rho}{2}}^k=\frac{2\tau(A-KC,\rho)}{1-\rho}
	\end{align*}
	This further implies
	\begin{align*}
		\norm{\matr{{cc}\RR_{\hat{A}-\hat{K}\hat{C}}&-\RR_{\hat{A}-\hat{K}\hat{C}}\hat{K}}}_{\Hinf}&\le (1+\norm{K}_2+\epsilon_K)	\norm{\RR_{\hat{A}-\hat{K}\hat{C}}}_{\Hinf}\\
		&\le (1+\norm{K}_2+\epsilon_K)\frac{2\tau(A-KC,\rho)}{1-\rho}.
	\end{align*}
	
	Now let $\Phi_w=\RR_{\hat{A}-\hat{K}\hat{C}}$ and $\Phi_v=-\RR_{\hat{A}-\hat{K}\hat{C}}\hat{K}$. The proof follows from~Lemma~\ref{LEM_error_analysis} and the inequality
	\begin{align*}
	\norm{\matr{{cc}\mbf\Phi_w & \mbf\Phi_v}\matr{{cc}\Delta_A&\Delta_K\\  \Delta_C&0}\matr{{c}\mbf \RR_AK\\ I}}_{\Ht}&\le \norm{\matr{{cc}\mbf\Phi_w & \mbf\Phi_v}}_{\Hinf}\norm{\matr{{cc}\Delta_A&\Delta_K\\  \Delta_C&0}\matr{{c}\mbf \RR_AK\\ I}R^{1/2}}_{\Ht}\\
	&\le  \sqrt{3}\epsilon(1+\norm{K}_2+\epsilon_K)\frac{2\tau(A-KC,\rho)}{1-\rho}\norm{\matr{{c}\mbf \RR_AK\\ I}R^{1/2}}_{\Ht}
	\end{align*}
\hfill $\blacksquare$
\subsection*{Proof of Lemma~\ref{LEM_error_analysis}}
	It is sufficient to show that
	\[
	\mbf{x}-\mbf{\tilde x}=\set{\paren{\mbf\Phi_w\Delta_A+\mbf\Phi_v\Delta_C}\mbf\RR_AK+\mbf \Phi_w \Delta_K+\mbf\Phi_w\hat{K}+\mbf\Phi_v}\mbf e,
	\]
	then the result follows from the definition of $\Ht$ norm and the fact that $R^{-1/2}e$ is white noise with unit variance.
	
	In frequency domain, equations~\eqref{EQN_System},~\eqref{EQN_Def_Predictor_class} can be rewritten as
	\begin{equation*}
	(zI-\hat{A})\mbf{\tilde x}=\mbf{L}(\mbf{y}-\hat{C}\mbf{\tilde{x}}), \,(zI-A)\mbf{x}=K(\mbf{y}-C\mbf{x})
	\end{equation*}
	Subtracting the two equations yields:
	\[
	(zI-\hat{A}+\mbf L\hat{C})(\mbf x-\mbf{\tilde x})+(-\Delta_A-\mbf{L}\hat{C}+KC)\mbf x=(K-\mbf L)\mbf y
	\]
	Using the fact that $\mbf y=C\mbf x +\mbf e $, we obtain:
	\begin{equation*}
	(zI-\hat{A}+\mbf{L}\hat{C})(\mbf{ x}-\mbf{\tilde x})=(\Delta_A-\mbf{L}[C-\hat{C}])\mbf x+(K-\mbf{L})\mbf e.
	\end{equation*}
	Multiplying from the left by $\mbf \Phi_w$ and using the fact that $\mbf \Phi_v=-\mbf\Phi_w \mbf L$
	\begin{equation*}
	\mbf{x}-\mbf{\tilde x}=(\mbf \Phi_w\Delta_A+\mbf\Phi_v\Delta_C)\mbf x+(\mbf \Phi_wK+\mbf\Phi_v)\mbf e
	\end{equation*}
	The result follows from adding and subtracting $\mbf \Phi_w\hat{K} \mbf e$ and the fact that $\mbf x=\RR_A K \mbf e$.
\hfill $\blacksquare$
\subsection*{Proof of Theorem~\ref{THM_SLS_Guarantees}}
	
	\noindent	\textbf{Step a:} First we prove that when  optimization problem~\eqref{EQN_Robust_Optimization} is feasible, the the mean square error is bounded by:
			\begin{equation}\label{EQN_Robust_Upper_Bound}
			\tilde{J} \le \sqrt{3}\C\epsilon\norm{\matr{{c}\mbf \RR_AK\\ I}}_{\Hinf}\snorm{R^{1/2}}_2+\text{\normalfont opt}(\C)\snorm{R^{1/2}}_2.
			\end{equation}
	 Assume  that $(\mbf\Phi_w,\mbf\Phi_v)$ is an optimal solution to~\eqref{EQN_Robust_Optimization}. From Lemma~\ref{LEM_error_analysis}:
	\begin{align*}
	\tilde{J}\le& \sqrt{3}\epsilon\norm{\matr{{cc}\mbf\Phi_w & \mbf\Phi_v}}_{\Ht}\norm{\matr{{c}\mbf \RR_AK\\ I}}_{\Hinf}\snorm{R^{1/2}}_2+\norm{(\mbf \Phi_w \hat{K}+\mbf \Phi_v)}_{\Ht}\snorm{R^{1/2}}_2,\\
	\le &\sqrt{3}\C\epsilon\norm{\matr{{c}\mbf \RR_AK\\ I}}_{\Hinf}\snorm{R^{1/2}}_2+\text{\normalfont opt}(\C)\snorm{R^{1/2}}_2,
	\end{align*}
	where we used $\norm{\matr{{cc}\mbf\Phi_w & \mbf\Phi_v}}_{\Ht}\le \C$ and optimality of $(\mbf\Phi_w,\mbf\Phi_v)$.

	\noindent	\textbf{Step b:}
	We prove that under condition~\eqref{EQN_Robustness_Condition}, the static Kalman gain $K$ is a feasible gain for~\eqref{EQN_Robust_Optimization}; equivalently, the responses $\mbf{\tilde{\Phi}}_{w}=\RR_{\hat{A}-K\hat{C}}$, and $\mbf{\tilde{\Phi}}_{v}=-\RR_{\hat{A}-K\hat{C}}K$ satisfy the constraints of~\eqref{EQN_Robust_Optimization}. 
	Consider the responses $\mbf \Phi_{w,opt}\triangleq \RR_{A-KC}$ and $\mbf \Phi_{v,opt}\triangleq-\RR_{A-KC}K$, which are optimal for the original unknown system. They satisfy the affine relation for the original system:
	\[
	\matr{{cc}\mbf \Phi_{w,opt}& \mbf \Phi_{v,opt}}\matr{{c}zI-A\\-C}=I
	\]
	Adding and subtracting the estimated matrices, we can show that they also satisfy a perturbed affine relation for the estimated system:
	\[
	\matr{{cc}\mbf \Phi_{w,opt}& \mbf \Phi_{v,opt}}\matr{{c}zI-\hat{A}\\-\hat{C}}=I+\underbrace{(\mbf \Phi_{w,opt}\delta_A+\mbf \Phi_{v,opt}\delta_C)}_{\mbf \Delta}
	\]
If the perturbation $(I+\mbf \Delta)^{-1}$ is stable, we can multiply both sides from the left, which yields:
	\[
	\matr{{cc}\mbf{\tilde{\Phi}}_{w}& \mbf{\tilde{\Phi}}_{v}}\matr{{c}zI-\hat{A}\\-\hat{C}}=I,
	\]
	where we used the fact that:
	\[
	(I+\mbf \Delta)^{-1}\mbf \Phi_{w,opt}=\mbf{\tilde{\Phi}}_{w},\quad	(I+\mbf \Delta)^{-1}\mbf \Phi_{v,opt}=\mbf{\tilde{\Phi}}_{v}
	\]
	Under condition~\eqref{EQN_Robustness_Condition}, the perturbation $\mbf \Delta$ has norm bounded by:
	\[\norm{\mbf \Delta}_{\Hinf}\le (\epsilon_A+\epsilon_C\norm{K}_2)\norm{\RR_{A-KC}}_{\Hinf}\le 1/2\]
	Hence:
	\[
	\norm{(I+\mbf \Delta)^{-1}}_{\Hinf}\le \sum_{t=0}^{\infty}\snorm{\mbf \Delta}^t_{\Hinf}\le\frac{1}{1-\norm{\mbf \Delta}_{\Hinf}}=2
	\]
	which shows that the responses $\mbf{\tilde{\Phi}}_{w},\mbf{\tilde{\Phi}}_{v}$ are stable. By construction, they are also strictly proper.
	What remains to show is that the robustness constraint holds. We have:
	\begin{align*}
	\norm{\matr{{cc}\mbf{\tilde{\Phi}}_{w}& \mbf{\tilde{\Phi}}_{v}}}_{\Ht}&\le \norm{(I+\mbf \Delta)^{-1}\matr{{cc}\mbf \Phi_{w}& \mbf \Phi_{v}}}_{\Ht}\\
	&\le \norm{(I+\mbf \Delta)^{-1}}_{\Hinf}(1+\norm{K}_2)\norm{\RR_{A-KC}}_{\Ht}\\
	&\le 2 (1+\norm{K}_2)\norm{\RR_{A-KC}}_{\Ht}\le \C
	\end{align*}
	
	\noindent	\textbf{Step c:}	Since $K$ is a feasible gain, by suboptimality
	\begin{align*}
	\text{\normalfont opt}(\C)&\le \norm{\mbf{\tilde{\Phi}}_{w}\hat{K}+\mbf{\tilde{\Phi}}_{v}}_{\Ht}\le \norm{(I+\mbf\Delta)^{-1}}_{\Hinf}\norm{\mbf\Phi_{w}\hat{K}+\mbf\Phi_{v}}_{\Ht}\\
	&\le 2\norm{\mbf\Phi_{w}\hat{K}+\mbf\Phi_{v}}_{\Ht}=2\norm{\mbf\Phi_{w}(\hat{K}-K)}_{\Ht}\\
	&\le 2\epsilon\norm{\mbf\RR_{A-KC}}_{\Ht}
	\end{align*}
	where we used $\mbf\Phi_{v}=-\mbf\Phi_{w}K$.
\hfill $\blacksquare$

\section{Identification algorithm and analysis}
Here we briefly present the results from~\cite{tsiamis2019finite}.
The stochastic identification algorithm involves two steps.
First, we regress future outputs to past outputs to obtain a Hankel-like matrix, which is a product of an observability and a controllability matrix. Second, we perform a realization step, similar to the Ho-Kalman algorithm, to obtain estimates for $A,C,K$. The outline can be found in Algorithm~\ref{ALG_identification}

\begin{algorithm}[!t]
	\caption{Stochastic Identification Algorithm}
	\label{ALG_identification}
	\begin{algorithmic}[1] {}
		\Require $p,f$, $y_0,\dots,y_{N+p+f-1}$, $W$.
		\Ensure Estimates: $\hat{C}$, $\hat{A}$, $\hat{K}$. 
		\State Compute $\hat{G}=\sum_{k=p}^{N+p-1}Y^{+}_kY^{-*}_k\paren{\sum_{k=p}^{N+p-1}Y^{-}_kY^{-*}_k}^{-1}$.
		\State Compute SVD: $\hat{G}W=\matr{{cc}\hat{U}_1&\hat{U}_2}\matr{{cc}\hat{\Sigma}_1&0\\0&\hat{\Sigma}_2}\matr{{c}\hat{V}^*_1\\\hat{V}^*_2}$, $\hat{\Sigma}_1\in\R^{n\times n}$.
		\State Set $\hat{\O}_f=\hat{U}_1\hat{\Sigma}^{1/2}_1$, $\hat{\K}_f=\hat{\Sigma}^{1/2}_1\hat{V}^*_1W^{-1}$.
		\State Set $\hat{C}=\hat{\O}_{f}(1:m,:)$,  $\hat{K}=\hat{\K}_{p}(:,m(p-1)+1:mp)$.
		\State Set $\hat{A}=\hat{\O}_{f}(:,1:m(f-1))^{\dagger}\hat{\O}_{f}(:,m+1:mf)$
	\end{algorithmic}
\end{algorithm}

\textbf{Definitions.} Let $p,f$, with $p,f\ge n$ be two design parameters that define the horizons of the past and the future respectively. Assume that we are given $N+p+f-1$ output samples. We define the future outputs $Y^{+}_{k}\in \R^{mf}$ and past outputs $Y^{-}_k\in\R^{mp}$ at time $k\ge p$ as follows:
\begin{align*}
Y^{+}_{k}&\triangleq \matr{{c}y_{k}\\\vdots\\y_{k+f-1}},\quad
Y^{-}_{k}\triangleq\matr{{c}y_{k-p}\\\vdots\\y_{k-1}},\,k\ge p
\end{align*}
The past and future noises $E^{+}_k,E^{-}_k$
are defined similarly:
\begin{align*}
E^{+}_{k}&\triangleq \matr{{c}e_{k}\\\vdots\\e_{k+f-1}},\quad
E^{-}_{k}\triangleq\matr{{c}e_{k-p}\\\vdots\\e_{k-1}},\,k\ge p
\end{align*}
The (extended) observability matrix $\O_k\in \R^{mk\times n}$ and the reversed (extended) controllability matrix $\K_k\in \R^{n\times mk}$ are defined as:
\begin{equation}\label{EQN_Observability_Matrix}
\O_k\triangleq\matr{{cccc}C^*&A^*C^*&\cdots&(A^*)^{k-1}C^*}^*,
\end{equation}
\begin{equation}
\label{EQN_Kalman_Controllability}
\K_k\triangleq\matr{{cccc}(A-KC)^{k-1}K&\dots&(A-KC)K&K}
\end{equation}
respectively.
We define the Hankel matrix:
\begin{equation}\label{EQN_Hankel}
G\triangleq\O_f \K_p.
\end{equation}
Finally, for any $s\ge 2$, define block-Toeplitz matrix:
\begin{equation}
\label{POE_EQN_Innovation_Toeplitz}
\T_s\triangleq\matr{{cccc}I_m&0& &0\\CK&I_m&\cdots&0\\ \vdots&\vdots& &\vdots \\CA^{s-2}K&CA^{s-3}K&\cdots&I_m}. 
\end{equation}
Finally, define the covariance matrices of the (weighted) past and future noises:
\begin{align*}
\Sigma_{E,f}&\triangleq\mathbb{E}\paren{\T_f E^{+}_{k}E^{+*}_{k}\T^*_f }\\
\Sigma_{E,p}&\triangleq\mathbb{E}\paren{\T_p E^{-}_{k}E^{-*}_{k}\T^*_p}.
\end{align*}

\subsection{Regression step}
It can be shown that the future and past outputs satisfy the following relation:
\begin{equation}\label{EQN_basic_regression}
Y^{+}_k=GY^{-}_k + \O_f(A-KC)^px_{k-p}+ \T_f E^{+}_k.
\end{equation}
Based on~\eqref{EQN_basic_regression}, we compute the least squares estimate \begin{equation}\label{EQN_G}
\hat{G}=\sum_{k=p}^{N+p-1}Y^{+}_kY^{-*}_k\paren{\sum_{k=p}^{N+p-1}Y^{-}_kY^{-*}_k}^{-1}.
\end{equation} 
The next theorem analyzes the sample complexity of the regression step~\citep{tsiamis2019finite}.
\begin{theorem}[Regression Step Analysis]\label{MAIN_THM_Hankel_Matrix}
	Consider system~\eqref{EQN_System}.
	Let $\hat{G}$, with $p\ge f$, be the estimate~\eqref{EQN_G} of the subspace identification algorithm given an output trajectory $y_0,\dots,y_{N+p+f-1}$ and let $G$ be as in~\eqref{EQN_Hankel}. Fix a confidence $\delta>0$ and define:
	\begin{equation}\label{MAIN_EQN_DELTA_N}
	\delta_N\triangleq\paren{2(N+p-1)m}^{-\log^2\paren{2pm}\log\paren{2(N+p-1)m}}.
	\end{equation}
	If $N\ge  \mathrm{poly}(\log 1/\delta,p,\T_p)$, then with probability at least $1-\delta_N-6\delta$:
		\begin{align} \label{MAIN_EQN_Basic_Theorem_Statement}
		\norm{G-\hat{G}}_2\le & \underbrace{8\mathcal{C}_1\sqrt{\frac{fmp}{N}\log\frac{5f\kappa}{\delta}}}_{O\paren{\sqrt{p\log N/N}}} +\C_2\underbrace{\norm{\paren{A-KC}^p}_2}_{O\paren{\rho(A-KC)^p}},
		\end{align}
	where 
	\begin{equation}\label{MAIN_EQN_Condition_Number_Approx}
	\kappa=\frac{4}{\sigma_{\min}\paren{\Sigma_{E,p}}}\paren{\norm{\O_p}^2_2\Tr\Gamma+\Tr\Sigma_{E,p}}+\delta
	\end{equation}	
	over-approximates the condition number of $\mathbb{E}\paren{Y_-Y_-^*}$, $\Gamma$ is the steady-state covariance matrix $\Gamma\triangleq \lim_{k\rightarrow \infty}\mathbb{E}x_kx^*_k$, and
	\begin{equation}\label{MAIN_EQN_System_Specific_Constant}
	\C_1= \sqrt{\frac{\norm{\Sigma_{E,f}}_2}{\sigma_{\min}\paren{\Sigma_{E,p}}}},\quad \C_2=4\norm{\O_f}_2\snorm{\O^{\dagger}_p}_2
	\end{equation}
are system-dependent constants, where $\dagger$ denotes the pseudo-inverse.
	\hfill $\diamond$
\end{theorem}
\begin{proof}
	We follow exactly the proof from~\cite{tsiamis2019finite}, but we improve a constant by avoiding applying the sub-multiplicative property of norm. The term $\norm{\Sigma_{E,f}}_2$ appears instead of its upper bound $\sqrt{\norm{R}_2}\norm{\T_f}_2$. 
\end{proof}
The above bounds depend polynomially on $\T_p,\T_f$. If $A$ has simple eigenvalues (or  Jordan blocks of small size), then the bounds depend polynomially on $n,m,f$ as well. However, if $A$ has a large Jordan block, e.g. of size $O(n)$, it is possible that the bounds scale exponentially with $n$; this follows from the fact that  $\snorm{\T_k}_2$ scales with $\max_{0\le i\le k-1}\snorm{A^i}_2$ in the worst case. 

To recover consistency, we need to make term $\C_2\norm{(A-KC)^p}_2$ go to zero at least as fast as $1/\sqrt{N}$. Since $A-KC$ is stable, it is sufficient to select $p=\beta \log N$, for some large enough $\beta$. The following corollary follows directly from Theorem~\ref{MAIN_THM_Hankel_Matrix}.

\begin{corollary}[Consistency]\label{COR_consistency}
	Consider the conditions of Theorem~\ref*{MAIN_THM_Hankel_Matrix} and the definition of $\delta_N$ in~\eqref{MAIN_EQN_DELTA_N}.
 Fix a confidence $\delta>0$ and let $\rho>\rho(A-KC)$. Select 
 \begin{equation}\label{EQN_beta}
 p=\beta\log N,\,\beta>-1/2\frac{1}{\log \rho}
 \end{equation}
If $N\ge  \mathrm{poly}(\log 1/\delta,\beta)$, then with probability at least $1-\delta_N-6\delta$:
	\begin{align*}
	\norm{G-\hat{G}}_2\le  \sqrt{\frac{\norm{\Sigma_{E,f}}_2}{\sigma_{\min}\paren{\Sigma_{E,p}}}}\sqrt{fmp}\tilde{O}\paren{\sqrt{\frac{\log 1/\delta}{N}}},
	\end{align*}
where 
$\tilde{O}$ hides logarithmic terms of $N$, constants, and other system parameters.
\hfill $\diamond$
\end{corollary}
Condition $\beta>-1/2\frac{1}{\log \rho}$ guarantees that $\norm{(A-KC)^p}_2=O(1/\sqrt{N})$.
\subsection{Realization step}\label{Section_SVD}
First, we compute a rank-$n$ factorization of the full rank matrix $\hat{G}W$, where $W$ is a user choice.
Let the SVD of $\hat{G}W$ be:
\begin{equation}\label{EQN_SVD}
\hat{G}W=\matr{{cc}\hat{U}_1&\hat{U}_2}\matr{{cc}\hat{\Sigma}_1&0\\0&\hat{\Sigma}_2}\matr{{c}\hat{V}_1^*\\\hat{V}_2^*},
\end{equation}
where $\hat{\Sigma}_1\in\R^{n\times n}$ contains the $n-$largest singular values.  
Then, a standard realization of $\O_f$, $\K_p$ is:
\begin{equation}\label{EQN_realization_observability_controllability}
\hat{\O}_{f}=\hat{U}_1\hat{\Sigma}^{1/2}_1,\: \hat{\K}_{p}=\hat{\Sigma}^{1/2}_1\hat{V}^*_1W^{-1}_2.
\end{equation}
For the theoretical finite sample bounds we used $W=I$, but in simulations the choice $W=\paren{\sum_{k=p}^{N+p-1}Y^{-}_kY^{-*}_k}^{1/2}$ works better--see for example MOESP~\citep{qin2006overview}.
For this step, we need the following assumption.
\begin{assumption}\label{ASS_SVD}
The order $n$ of the system is known. The pair $(A,K)$ is controllable.
\end{assumption}

Based on the estimated observability/controllability matrices, we can approximate the system parameters as follows:
\begin{equation*}\label{EQN_C_K_estimate}
\hat{C}=\hat{\O}_f\paren{1:m,:},\quad\hat{K}=\hat{\K}_p\paren{:,(p-1)m+1:pm},
\end{equation*} 
where the notation $\hat{\O}_f\paren{1:m,:}$ means we pick the first $m$ rows and all columns. The notation for $\hat{\K}_p$ has similar interpretation.
For simplicity, define 
\[\hat{\O}^{u}_{f}\triangleq \hat{\O}_{f}\paren{1:m(f-1),:},
\]
which includes the $m(f-1)$ ``upper" rows of matrix $\hat{\O}_{f}$. Similarly, we define the lower part $\hat{\O}^{l}_{f}$.
For matrix $A$ we exploit the structure of the extended observability matrix and solve
$
\hat{\O}^{u}_{f}\hat{A}=\hat{\O}^{l}_p
$
in the least squares sense by computing
\begin{equation*}\label{EQN_A_estimate}
\hat{A}=\paren{\hat{\O}^{u}_{f}}^{\dagger}\hat{\O}^{l}_p,
\end{equation*} 
where $\dagger$ denotes the pseudoinverse.

The next theorem analyzes the robustness of the realization step. Before we state it, let us introduce some notation. Assume that we knew $G$ exactly. Then, the SVD in the realization step would be:
\begin{equation*}
GW=\matr{{cc}U_1&U_2}\matr{{cc}\Sigma_1&0\\0&0}\matr{{c}V_1^*\\V_2^*}=U_1\Sigma_1V^*_1,
\end{equation*}
for some $\Sigma_1\in \R^{n\times n}$.
Hence, if we knew $GW$ exactly, the estimated observability and controllability matrices would be
\begin{equation}
\bar{\O}_f=U_1\Sigma^{1/2}_1,\,\bar{\K}_p=\Sigma^{1/2}_1V^*_1W^{-1}.
\end{equation}
The original matrices $\O_f,\K_p$ and $\bar{\O}_f,\bar{\K}_p$ are equivalent up to the similarity transformation $\O_fS=\bar{\O}_f$, $S^{-1}\K_p=\bar{\K}_p$ where
\begin{equation}\label{EQN_Similarity}
S\triangleq \O^{\dagger}_f\bar{\O}_f.
\end{equation}
\begin{theorem}[Realization robustness]\label{SVD_THM_Main}
	Suppose that Assumption~\ref{ASS_SVD} holds. Consider the true Hankel-like matrix $G$ defined in~\eqref{EQN_Hankel} and the noisy estimate $\hat{G}$ defined in~\eqref{EQN_G}, with $p,f>n$. Let $\hat{A},\hat{C},\hat{K},\hat{\O}_f$, $\hat{\K}_p$ be the output of the balanced realization algorithm based on $\hat{G}$ with $W=I$. Let $S$ be the similarity transformation~\eqref{EQN_Similarity}. If $G$ has rank $n$ and the following robustness condition is satisfied:
	\begin{equation}\label{SVD_EQN_robustness_condition}
	\norm{\hat{G}-G}_2\le \frac{\sigma_{n}\paren{G}}{4},
	\end{equation}
	then there exists an orthonormal matrix $T\in\R^{n\times n}$ such that:
	\begin{align}\label{SVD_EQN_Bound_OBS_CONT}
	&	\norm{\hat{\O}_f-\O_fST}_{2}\le 2\sqrt{\frac{10n}{\sigma_{n}\paren{G}}}\norm{G-\hat{G}}_2\nonumber\\
	&	\norm{\hat{C}-CST}_{2}\le \norm{\hat{\O}_f-\O_fST}_{2}\nonumber\\
	&	\norm{\hat{A}-T^*S^{-1}AST}_{2}\le \underbrace{\frac{\sqrt{\norm{G}_2}+\sigma_o}{\sigma^2_o}}_{O\paren{1}}\norm{\hat{\O}_f-\O_fST}_{2}\nonumber\\
	&	\norm{\hat{K}-T^*S^{-1}K}_{2}\le 2\sqrt{\frac{10n}{\sigma_{n}\paren{G}}}\norm{G-\hat{G}}_2,\nonumber
	\end{align}
	where
	$
	\sigma_o=\min\paren{\sigma_{n}\paren{\hat{\O}^{u}_{f}},\sigma_{n}\paren{\O^{u}_{f}S}}.
	$
	The notation $\hat{\O}^{u}_{f},\O^{u}_{f},$ refers to the upper part of the respective matrix (first $(f-1)m$ rows).
	\hfill $\diamond$
\end{theorem}

\section{Formal end-to-end result}\label{Section_formal}
If we combine Corollary~\ref{COR_consistency} and Theorem~\ref{SVD_THM_Main}, we obtain finite sample guarantees for the estimation of matrices $A,C,K$.  
Meanwhile, Theorems~\ref{THM_CE_performance} and~\ref{THM_SLS_Guarantees} provide two different solutions to Problem~\ref{Prob_Synthesis}. 
Putting everything together gives the following formal end-to-end bound.

\begin{theorem}[End-to-end guarantees]\label{THM_EndToEnd_Formal}
	Consider the conditions of Theorem~\ref{MAIN_THM_Hankel_Matrix} and suppose Assumption~\ref{ASS_SVD} holds. Let $p=\beta \log N$, $p\ge f>n$, with $\beta$ as in~\eqref{EQN_beta}. Consider the definition of $S$ in~\eqref{EQN_Similarity} and $\delta_N$ in~\eqref{MAIN_EQN_DELTA_N}. Fix a failure probability $\delta \in (0,1)$. 
	Then, if \[N\ge \mathrm{poly}(\log(1/\delta),\beta,\sigma_n(G)),\]  with probability at least $1-6\delta-\delta_N$ the identification and filter synthesis pipeline of Fig.~\ref{Figure_CoarseID_architecture}, with system identification performed as in Algorithm~\ref{ALG_identification} with $W=I$ and filter synthesis performed as in Sections~\ref{Section_CE},~\ref{Section_Dynamic}, achieves mean squared prediction error satisfying
	\begin{align}\label{EQN_Formal_Result}
\sqrt{\lim_{t\rightarrow \infty}\frac{1}{t}\sum_{k=0}^{t}\snorm{\tilde{x}_k-T^*S^{-1}x_k}_2}\le \C_{ID}\C_{KF}\tilde{O}(\sqrt{\frac{\log 1/\delta}{N}})
	\end{align}
	for some orthonormal matrix $T$.
	Constant $\C_{KF}$ is defined as:
		\begin{align*}
\C_{KF}=\inf_{\rho>\rho(A-KC)}\frac{\tau(A-KC,\rho)}{1-\rho}(1+\snorm{K}_2)\norm{\matr{{c} \RR_AK\\ I}R^{1/2}}_{\Ht}
	\end{align*}
	in the case of CE Kalman filtering and
			\begin{align*}
	\C_{KF}=\norm{\RR_{A-KC}}_{\Ht}(1+\snorm{K}_2)\norm{\matr{{c} \RR_AK\\ I}R^{1/2}}_{\Hinf}
	\end{align*}
	in the case of robust KF.
Constant $\C_{ID}$ captures the difficulty of identifying system~\eqref{EQN_System} and is defined as:
	\begin{equation}\label{EQN_Cid}
\C_{ID}=\sqrt{\frac{\norm{\Sigma_{E,f}}_2}{\sigma_{\min}\paren{\Sigma_{E,p}}}}\frac{1}{\sigma_{n}(\O_{f-1}S)\sqrt{\sigma_n(G)}}\sqrt{fmpn}
\end{equation}
  Here, $\tilde{O}$ hides constants, other system parameters, and logarithmic terms.
\end{theorem}
The intuition behind constant $\C_{ID}$ is the following. The noise both excites the system and also introduces errors that obstruct identification; this is captured by the square root of the condition number of the covariances $\Sigma_{E,f},\Sigma_{E,p}$. Moreover, $\sigma_{n}(\O_{f-1}S)$ quantifies how easy it is to observe system~\eqref{EQN_System}. A similar interpretation holds for $\sigma_n(G)$. Finally, larger dimensions $f,p,m,n$ require more samples for identification since there are more unknowns in matrix $G$.

Note that the mean squared prediction error in~\eqref{EQN_Formal_Result} is computed with respect to the estimated state-space basis, i.e. up to the similarity transformation $ST$, where $S$ is defined in~\eqref{EQN_Similarity} and $T$ is some orthonormal matrix.
In terms of the original state-space basis, the mean squared prediction error~\eqref{EQN_Formal_Result} would be:
\[
\sqrt{\lim_{t\rightarrow \infty}\frac{1}{t}\sum_{k=0}^{t}\snorm{ST\tilde{x}_k-x_k}_2}\le \norm{S}_2\C_{ID}\C_{KF}\tilde{O}(\sqrt{\frac{\log 1/\delta}{N}})
\]
From~\eqref{EQN_Similarity}, the norms of $S,S^{-1}$ are bounded, so, the bound~\eqref{EQN_Formal_Result} is not vacuous.

\end{document}